\documentclass[aps, preprint]{revtex4-2}

\usepackage[utf8]{inputenc}
\usepackage[T1]{fontenc}
\usepackage{amsmath}
\usepackage{amsthm}
\usepackage{amssymb}
\usepackage{hyperref}
\usepackage{dsfont}
\usepackage{graphicx}
\usepackage{enumerate}
\usepackage{tensor}

\hypersetup{colorlinks=true, linkcolor=blue, citecolor=blue}

\newcommand{\Real}{\mathbb{R}}

\newcommand{\Comp}{\mathbb{C}}
\newcommand{\seq}{\subseteq}
\newcommand{\pspace}{\Real^3 \setminus \{0\}}

\newcommand{\Lightcone}{\mathcal{L}_+}

\newcommand{\Poincare}{Poincar\'{e} }
\newcommand{\SO}{\mathrm{SO}}
\newcommand{\so}{\mathfrak{so}}

\newcommand{\ISO}{\mathrm{ISO}}

\newcommand{\khat}{\hat{\boldsymbol{k}}}
\newcommand{\vhat}{\hat{\boldsymbol{v}}}

\newcommand{\mbf}[1]{\boldsymbol{#1}}
\newcommand{\DR}{\cx^{R}}
\newcommand{\DK}{\cx^{K}}

\newcommand{\Df}{\cx^{f}}
\newcommand{\Dpm}{\cx^{\pm}}
\newcommand{\Dp}{\cx^{+}}
\newcommand{\Dm}{\cx^{-}}
\newcommand{\barDK}{\bar{\mbf{D}}^{K}}
\newcommand{\barDR}{\bar{\mbf{D}}^{R}}
\newcommand{\QR}{\mbf{Q}^{R}}
\newcommand{\QK}{\mbf{Q}^{K}}

\newcommand{\cx}{\boldsymbol{D}} 
\newcommand{\vf}[1]{\boldsymbol{#1}} 
\newcommand{\secip}[2]{\langle #1 , #2 \rangle_{s}} 
\newcommand{\kvec}{\mbf{k}} 
\newcommand{\ie}{\emph{i.e.}} 
\newcommand{\uv}[1]{\mbf{#1}} 
\newcommand{\phat}{\hat{\mbf{P}}}
\newcommand{\kmag}{|\mbf{k}|}
\newcommand{\ek}{\uv{e}_k}
\newcommand{\etheta}{\uv{e}_\theta}
\newcommand{\ephi}{\uv{e}_\phi}

\newcommand{\LD}{L^{\mbf{D}}}
\newcommand{\LDbold}{\mbf{L}^{\cx}}
\newcommand{\SD}{S^{\cx}}
\newcommand{\SDbold}{\mbf{S}^{\cx}}
\newcommand{\sop}[1]{#1} 
\newcommand{\svop}[1]{\mbf{#1}} 
\newcommand{\phats}{\hat{\mbf{P}}}

\newtheorem{theorem}{Theorem}

\newtheorem{definition}{Definition}
\newtheorem{proposition}{Proposition}

\begin{document}

\title{Zero curvature is a necessary and sufficient condition for a spin-orbital decomposition}
\date{\today}

\author{Eric Palmerduca}
\email{ep11@princeton.edu}
\affiliation{Department of Astrophysical Sciences, Princeton University, Princeton, New Jersey 08544, USA}
\affiliation{Princeton Plasma Physics Laboratory, Princeton, NJ 08543,
USA}

\author{Hong Qin}
\email{hongqin@princeton.edu}
\affiliation{Department of Astrophysical Sciences, Princeton University, Princeton, New Jersey 08544, USA}
\affiliation{Princeton Plasma Physics Laboratory, Princeton, NJ 08543,
USA}

\begin{abstract}
There has been an extended debate regarding the existence of a spin-orbital decomposition of the angular momentum of photons and other massless particles. It was recently shown that there are both geometric and topological obstructions preventing any such decomposition. Here we show that any geometric connection on a particle's state space induces a splitting of the angular momentum into two operators. These operators are well-defined angular momentum operators if and only if the connection has zero curvature. Massive particles have two canonical curved connections corresponding to boosts and rotations, respectively. These can be uniquely combined to produce a flat connection, and this gives a novel derivation of the Newton-Wigner position operator and the corresponding spin and orbital angular momenta for relativistic massive particles. When the mass is taken to zero, transverse boosts and rotations degenerate, leaving only a single connection for massless particles. This connection produces a commonly proposed splitting of the massless angular momentum into two operators. However, the connection is not flat, explaining why these operators do not satisfy the angular momentum commutation relations and are thus not true spin and orbital angular momentum operators.
\end{abstract}

\maketitle

\section{Introduction}

There has been an extended controversy \cite{Akhiezer1965,Jaffe1990,VanEnk1994_EPL_1,VanEnk1994_JMO_2,Chen2008,Wakamatsu2010,Bliokh2010,Bialynicki-Birula2011,Leader2013,Leader2014,Leader2016,Leader2019,Yang2022, PalmerducaQin_PT, Das2024, PalmerducaQin_GT, PalmerducaQin_SAMOAM} surrounding a basic question: can the angular momentum operator of massless particles, such as photons and gravitons, be split into well-defined spin and orbital components? This can of course be done for massive particles, but issues immediately arise in the massless case. For massive particles, the internal (polarization) and external (momentum) degrees of freedom (DOFs) are independently rotationally symmetric, and the generators of these $\SO(3)$ symmetries are the spin (SAM) and orbital angular momentum (OAM) operators, respectively \cite{Terno2003}. For massless particles, the situation is different, as the internal and external DOFs are coupled \cite{PalmerducaQin_SAMOAM}. For example, a photon's momentum $\kvec$ and polarization $\mbf{E}$ must satisfy the transversality constraint $\kvec \cdot \mbf{E} = 0$. $\kvec$ and $\mbf{E}$ can be simultaneously rotated, and the generator of this symmetry is the total angular momentum operator $\mbf{J} = (J_1,J_2,J_3)$. However, $\kvec$ and $\mbf{E}$ cannot be independently rotated without violating the transversality constraint, that is, without creating unphysical states. This shows that the most straightforward attempt at an SAM-OAM decomposition fails. It has been suggested that instead the decomposition for massless particles should be \cite{Bliokh2010,Bialynicki-Birula2011}
\begin{subequations}\label{eq:massless_splitting}
\begin{align}
    \mbf{J} &= \mbf{J}_\parallel + \mbf{J}_\perp \label{eq:helicity_decomp} \\
    \mbf{J}_\parallel &= (\mbf{J}\cdot \phat) \phat = \chi\phat \label{eq:spin_attempt}\\
    \svop{J}_\perp &= -\frac{1}{H}\mbf{P} \times \mbf{K}  \label{eq:orbital_attempt}
\end{align}
\end{subequations}
with $\mbf{J}_\parallel$ and $\mbf{J}_\perp$ being the SAM and OAM operators, respectively. Here $\mbf{P}$ is the momentum operator, $\mbf{K}$ is the generator of boosts, $H$ is the Hamiltonian, $\phat \doteq \mbf{P}/|\mbf{P}|$, and $\chi \doteq \mbf{J}\cdot \phat $ is the helicity operator. $\mbf{J}_\parallel$ and $\mbf{J}_\perp$ are well-defined vector operators, that is,
\begin{align}
    [J_{\parallel,a},J_b] &= i\epsilon_{abc}J_{\parallel,c} \\
    [J_{\perp,a},J_b] &= i\epsilon_{abc}J_{\perp,c}.
\end{align}
Furthermore, these operators are defined purely in terms of \Poincare generators, so the expression is coordinate independent. However, there is the fundamental issue that $\mbf{J}_\parallel$ and $\mbf{J}_\perp$ are not actually angular momentum operators. Indeed, they do not generate rotational $\SO(3)$ symmetries, or equivalently, they do not satisfy the $\so(3)$ angular momentum commutation relations:
\begin{gather}
    [J_{\parallel,a},J_{\parallel_b}] = 0 \neq i\epsilon_{abc}J_{\parallel,c} \\
    [J_{\perp,a}, J_{\perp,b}] = i\epsilon_{abc}(J_{\perp,c} - J_{\parallel,c})\neq i\epsilon_{abc}J_{\perp,c}.\label{eq:J_perp_comm}
\end{gather}
This issue was first observed by van Enk and Nienhuis \cite{VanEnk1994_EPL_1}. It was recently shown that it is in fact impossible to produce a genuine SAM-OAM splitting for massless particles, as the internal and external spaces do not permit any $\SO(3)$ symmetries \cite{PalmerducaQin_SAMOAM}. Any attempted splitting will result in operators which are either not gauge invariant or else fail to satisfy angular momentum commutation relations. Given that $\mbf{J}_\parallel$ and $\mbf{J}_\perp$ fail to be angular momentum operators, this splitting appears ad hoc. We note that the helicity operator $\chi = \mbf{J} \cdot \phat$ has a clear theoretically significance. Indeed, $\chi$ is the massless analog of the massive spin operator $\mbf{S}$ in Wigner's classification of particles, generating the massless little group $\SO(2)$ \footnote{Technically, the massless little group is the two dimensional Euclidean group $\ISO(2)$. However, if one considers only particles with a finite number of internal DOFs (which includes all particles in the Standard Model), then the inhomogeneous elements of $\ISO(2)$ act by the identity \cite{Weinberg1995}, so that the effective action is that of $\SO(2)$.} in a similar way that $\mbf{S}$ generates the massive little group $\SO(3)$ \cite{Wigner1939,Weinberg1995,PalmerducaQin_PT,PalmerducaQin_SAMOAM}. However, it is the helicity $\chi$, not its ``vectorization'' $\mbf{J}_\parallel = \chi\phat$, which shows up in Wigner's method. One of the main goals of this article is to show that the massless $\mbf{J}_\parallel$, $\mbf{J}_\perp$ splitting does in fact result from a general geometric construction for particles, albeit one that is different from Wigner's little group construction.

To understand this, we recall some of the subtleties of the SAM-OAM splitting for relativistic massive particles. Wigner showed that massive particles are classified by their spin $s$ and that their internal polarization space has dimension $2s+1$ (neglecting DOFs describing non-spacetime symmetries, such as color charge) \cite{Wigner1939,Weinberg1995}. It is then possible to choose coordinates such that \cite{Terno2003}
\begin{equation}\label{eq:massive_splitting}
    \mbf{J} = -i\mbf{P}\times \mbf{Q}^{NW} + \mbf{S}.
\end{equation}
In this coordinate representation, $\mbf{J}$ acts on $2s+1$ component wave functions in $\kvec$ space, $\mbf{S}$ acts as the $(2s+1) \times (2s+1)$ spin matrices, and $\mbf{Q}^{NW}$ acts as $i\nabla$, the component-wise gradient in $\kvec$ space. It is important to note that this simple description of $\mbf{S}$ and $\mbf{Q}^{NW}$ is only valid in a very specific set of coordinates \cite{Jordan1980} and it hides considerable complexity. Indeed, the coordinate independent descriptions of $\mbf{S}$ and $\mbf{Q}^{NW}$ are much more complicated \cite{Jordan1980}:
\begin{align}
    \mbf{S} &= \frac{1}{m}(H\mbf{J} + \mbf{P}\times \mbf{K}) -\frac{1}{m(H+m)}(\mbf{P}\cdot\mbf{J})\mbf{P} \\
    \mbf{Q}^{NW} &= \frac{1}{H}\Big(\mbf{K} - \frac{\mbf{P}}{2H}\Big) - \frac{1}{mH(H+m)}\mbf{P} \times(H\mbf{J} + \mbf{P}\times \mbf{K}) \label{eq:NW_Jordan}.
\end{align}
$\mbf{Q}^{NW}$ is known as the Newton-Wigner position operator. By choosing an internal basis of eigenvectors of $S_3$, one obtains the simpler coordinate representation in terms of spin matrices and gradients \cite{Jordan1980}. The operator $\mbf{Q}^{NW}$ was defined by Newton and Wigner in their search for a relativistic position operator for massless particles \cite{Newton1949}, and the coordinate independent form (\ref{eq:NW_Jordan}) was derived by Jordan \cite{Jordan1980} using the Pauli-Lubanski vector. In this article we will present a novel and physically intuitive construction of $\mbf{Q}^{NW}$ and the relativistic SAM-OAM splitting (\ref{eq:massive_splitting}) based on the theory of connections on vector bundles. What is particularly important about this construction is that it also produces the splitting (\ref{eq:massless_splitting}) when applied to massless particles. This shows that a single geometric procedure results in both the known massive and massless splittings, and particularly that the latter is not ad hoc. Furthermore, we will see the precise singularity that causes the massless operators to fail to be true angular momentum operators.

In Sec. \ref{sec:Formalism} we will review the vector bundle description of particles. In Sec. \ref{sec:Connection_induced_splitting} we show how connections on vector bundles induces splittings of $\mbf{J}$, but that only flat connections induce SAM-OAM decompositions. In Sec. \ref{sec:Applications} we explore specific examples of such splittings for elementary particles induced by \Poincare symmetry, showing that massive particles admit a unique SAM-OAM splitting which is obstructed when the mass is taken to zero.

\section{The vector bundle description of particles}\label{sec:Formalism}
We will use the vector bundle description of massless particles \cite{Simms1968, PalmerducaQin_PT,PalmerducaQin_GT,PalmerducaQin_helicity,PalmerducaQin_SAMOAM,Dragon2024} throughout this paper, particularly following the conventions of Refs.
\cite{PalmerducaQin_PT,PalmerducaQin_GT,PalmerducaQin_helicity,PalmerducaQin_SAMOAM}. A basis for the states of an elementary particle are given by the momentum eigenstates $(k,v)$ where $k^\mu = (\omega,\kvec)$ is the $4$-momentum and $v$ is the internal polarization vector \cite{Weinberg1995}. $\omega = \sqrt{m^2 + \kmag^2}$ for a mass $m$ particle, so $k^\mu$ is restricted to the mass $m$ hyperboloid $\mathcal{M}_m$ for massive particles or to the (forward) lightcone $\Lightcone$ for massless particles. In either case, the momentum can be labeled just by $\kvec$, so states can also be written as $(\kvec,v)$. Note that with this parameterization $\mathcal{M}_m \cong \Real^3$ and $\Lightcone \cong \pspace$; the origin is removed in the latter case because massless particles have no rest frame. The collection of all states $(\kvec_0,v)$ with fixed momentum $\kvec_0$ forms a vector spaces denoted $E(\kvec_0)$. Thus the collection of all states, denoted $E$, is a family of vector spaces parameterized by $M = \mathcal{M}_m$ or $M = \Lightcone$. Such a parameterized collection of vector spaces $\pi:E\rightarrow M$ is called a vector bundle with base manifold $M$. $\pi$ is the canonical projection $\pi(\kvec,v) = \kvec$ and the rank $r$ of the vector bundle is the dimension of the vector spaces $E(\kvec)$. We will follow the standard convention of referring to both $E$ and the triple $\pi:E\rightarrow M$ as a vector bundle. Vector bundles are one of the principle objects of study in topology and differential geometry, and thus this description of particles allows access to many established techniques from these fields. Additional information on vector bundle techniques can be found in Refs. \cite{Atiyah_K_Theory, Tu2017differential}. The vector space $E(\kvec_0)$ is called the fiber at $\kvec_0$. A section $\psi:M\rightarrow E$ is a choice of vector in each fiber, \ie, $\psi(\kvec) \in E(\kvec)$. A section describes a superposition of momentum eigenfunctions and is called a wave function in physics terminology. The space of wave functions is denoted $L^2(E)$. Each fiber $E(\kvec)$ is actually an inner product space, that is, $\langle (\kvec,v_1),(\kvec,v_2)\rangle$ is well-defined. This makes $E$ into a Hermitian vector bundle. This induces an inner product on $L^2(E)$:
\begin{equation}
    \secip{\psi_1}{\psi_2} = \int \frac{d^3k}{\omega} \langle \psi_1(\kvec),\psi_2(\kvec)\rangle
\end{equation}
where $\omega^{-1}d^3k$ is the Lorentz invariant measure. As suggested by the notation $L^2(E)$, we consider only square integrable sections with respect to this measure.

Elementary particles correspond to unitary irreducible representations (UIRs) of the (proper orthochronous) \Poincare group $\ISO^+(3,1)$, which consists of all combinations of boosts, rotations, and spacetime translations \cite{Wigner1939,Weinberg1995}. Thus, there is a unitary irreducible action $\Sigma$ on the states in $E$. For simplicity, we will restrict our discussion to bosons since our focus is on massless particles and there are no known massless fermions. Thus, we will not consider projective representations. For any $\Lambda \in \ISO^+(3,1)$, $\Sigma_\Lambda$ is a vector bundle isomorphism on $E$, linearly mapping a fiber $E(k)$ to $E(\Lambda k)$, where $\Lambda k$ denotes the standard \Poincare action on momentum 4-vectors. That is,
\begin{equation}
    \Sigma_\Lambda(k,v) = (\Lambda k, \Sigma_{\Lambda}v).
\end{equation}
$\Sigma$ is unitary in that it preserves the Hermitian product between vectors. It is also irreducible, meaning there are no proper subbundles of $E$ (with nonzero rank) which are also unitary representations. This makes $E$ into a vector bundle UIR of $\ISO^+(3,1)$. The bundle representation induces a corresponding vector space UIR $\tilde{\Sigma}$ on the space of wave function $L^2(E)$, namely, \cite{Simms1968,PalmerducaQin_PT}
\begin{equation}
    (\tilde{\Sigma}_{\Lambda}\psi)(\kvec) = \Sigma_{\Lambda}[\psi(\Sigma_{\Lambda^{-1}}\kvec)].
\end{equation}
The generators of $\tilde{\Sigma}$ are the Hamiltonian $\sop{H}$, momentum $\svop{P} = (\sop{P}_1, \sop{P}_2,\sop{P}_3)$, angular momentum $\svop{J}$, and boost $\svop{K}$ operators. In particular,
\begin{align}
    \sop{J}_a\psi &= i\frac{d}{dt}\Big|_{t=0} \tilde{\Sigma}_{R_a(t)}\psi \\
    \sop{K}_a\psi &= -i\frac{d}{dt}\Big|_{t=0} \tilde{\Sigma}_{\Lambda_a(t)}\psi \\
\end{align}
where $R_a(t) \in \SO(3) \seq\SO^+(3,1)$ is a rotation by $t$ about the $a$ axis and $\Lambda_a(t)$ is a boost in direction $\uv{e}_a$ by velocity $t$. Similarly, $\sop{H}$ and $\svop{P}$ are the generators of time and space translations and
\begin{align}
    \svop{P}\psi(\kvec) &= \kvec\psi(\kvec) \\
    \sop{H}\psi(\kvec) &= \sqrt{m^2 + \kmag^2}\psi(\kvec).
\end{align}
The generators satisfy the \Poincare algebra commutation relations given in Eq. (\ref{eq:Poincare_algebra}) in the  Appendix, and in particular, the angular momentum satisfies the $\so(3)$ relations
\begin{equation}
    [\sop{J}_a,\sop{J}_b] = i\epsilon_{abc}\sop{J}_c.
\end{equation}

Using Wigner's little group method \cite{Wigner1939,Weinberg1995}, the elementary particle bundles can be classified \cite{PalmerducaQin_PT,PalmerducaQin_helicity}.
The massive bundle representation are labeled by their spin $s$ and are isomorphic to the rank $2s+1$ trivial bundles $\pi:\mathcal{M}_m \times \Comp^{2s+1} \rightarrow\mathcal{M}_m$. Massless bosons, on the other hand, are all line bundles (rank 1 bundles) over the lightcone and are distinguished by a difference in topology. They are labeled by their integer helicity $h$ and are isomorphic to line bundles $\pi:\gamma_h\rightarrow \Lightcone \cong\pspace$. Line bundles over $\Lightcone$ are completely characterized by their (first) Chern number $C$, and the $\gamma_h$ can be abstractly described as the line bundles with $C(\gamma_h) = -2h$. $C$ can be considered as a measure of how topologically nontrivial or globally ``twisted'' the state space is \cite{Bott2013}. The $\gamma_h$ can also be explicitly constructed as tensor products of the $R$ and $L$ photon bundles, which are the bundles of solutions of Maxwell's equations in vacuum \cite{PalmerducaQin_helicity}. In this article we will only need the fact that massless particles with nonzero helicity have nonvanishing Chern number.

\section{Connection induced angular momentum decompositions}\label{sec:Connection_induced_splitting}
Our goal is to decompose the angular momentum operator into two vector operators
\begin{equation}
    \svop{J} = \svop{L} + \svop{S}.
\end{equation}
That $\svop{L}$ and $\svop{S}$ are vector operators just means that they rotate as vectors under $\SO(3) \seq \ISO^+(3,1)$, or equivalently, that they satisfy \cite{Hall2013}
\begin{align}
    [\sop{L}_a,\sop{J}_b] &= i\epsilon_{abc}\sop{L}_C \\
    [\sop{S}_a,\sop{J}_b] &= i\epsilon_{abc}\sop{S}_C.
\end{align}
Ideally, we would like $\svop{L}$ and $\svop{S}$ to be angular momentum operators, that is, to satisfy the $\so(3)$ angular momentum commutation relations
\begin{align}
    [\sop{L}_a,\sop{L}_b] &= i\epsilon_{abc}\sop{L}_c, \\
    [\sop{S}_a,\sop{S}_b] &= i\epsilon_{abc}\sop{S}_c.
\end{align}

In non-relativistic quantum mechanics, the orbital angular momentum operator is given in the momentum basis by $-i\kvec \times \nabla$ where $\nabla$ is the momentum space gradient on scalar functions. To generalize this to current setting of vector bundle representations of particles, we must find a suitable replacement for the operator $\nabla$. However, $\nabla$ is not uniquely defined on sections of vector bundles even when the base manifold $M$ is identified with an open subset of $\Real^3$ as in $\mathcal{M}_m \cong \Real^3$ or $\Lightcone \cong \pspace$ \cite{Tu2017differential}. The core issue is that there is not a unique way to identify the vectors at different base points in $M$ (even if one specifies a path on $M$ between the points). As such, there is ambiguity in how one measures the change in a wave function $\psi(\kvec) \in L^2(E)$ as $\kvec$ is varied. It turns out that there are many inequivalent ways to remove this ambiguity, and each describes a connection $\cx$ on the vector bundle $E$. Such connections can be thought of as covariant derivatives, and indeed, they generalize the affine connections that occur in general relativity; a detailed treatment of connections on vector bundles can be found in Ref. \cite{Tu2017differential}. The essential program will be to replace $\nabla$ with connections on vector bundles and study the resultant splittings of the angular momentum. Technically, a connection is defined as follows:
\begin{definition}[\cite{Tu2017differential}]
    Let $\mathfrak{X}(M)$ denote the space of vector fields $X:M\rightarrow TM$ on $M$. A connection $\cx$ on $\pi:E\rightarrow M$ is a map
    \begin{gather}
        \cx: \mathfrak{X}(M) \times L^2(E) \rightarrow L^2(E) \nonumber\\
        (X,\psi) \mapsto \cx_X \psi\nonumber
    \end{gather}
    such that for any smooth function $f:M\rightarrow \Comp$,  $\cx_X \psi$ is $f$-linear in $X$ and satisfies the Leibniz rule in $\psi$:
    \begin{gather}
        \cx_{fX}(\psi) = f\cx_X\psi \\
        \cx_X(f\psi) = df(X)\psi + f \cx_X \psi \label{eq:Leibniz}
    \end{gather}
    where $f:M\rightarrow \Comp$ is a smooth function and $df:TM\rightarrow \Comp$ is the differential of $f$. We also write $\cx_X$ as $\mbf{X}\cdot \mbf{D}$; this can be thought of as the derivative in direction $\mbf{X}$.
    
    In the case when $M$ is a an open subset of $\Real^3$, given $\vf{Y}\in \mathfrak{X}(M)$ we define the operator
    \begin{gather}
        \vf{Y} \times \cx:\mathfrak{X}(M) \times L^2(E) \rightarrow L^2(E)
    \end{gather}
    such that for $X \in \mathfrak{X}(M)$,
    \begin{equation}\label{eq:cross_product}
    \vf{X} \cdot (\vf{Y} \times \cx)\psi = [(\vf{X} \times  \vf{Y})\cdot \cx] \psi = \cx_{\vf{X}\times \vf{Y}}\psi.
    \end{equation}
\end{definition}
Given a vector field $\vf{X}(\kvec) = v^a(\kvec) \uv{e}_a$ on an open subset of $\Real^3$, where $\uv{e}_a$ are an orthogonal basis of $\mathbb{R}^3$, we can define the operator on $L^2(E)$
\begin{equation}
    \vf{X} \cdot \svop{K} = v^a\sop{K}_a
\end{equation}
and similarly for $\mbf{J}$ and $\mbf{P}$. In this sense, we can think of the generators $\svop{K}$, $\svop{J}$, and $\svop{P}$ as having the same domain and range as a connection:
\begin{equation}
    \svop{K}: \mathfrak{X}(M) \times L^2(E) \rightarrow L^2(E),
\end{equation}
with the Cartesian components $\sop{K}_a$ corresponding to the constant vector fields $\uv{e}_a$, $\sop{K}_a = \uv{e}_a \cdot \svop{K}$. Note also that
\begin{equation}
    [(\svop{P} \times \cx)\psi] = (\kvec \times \mbf{D})\psi
\end{equation}
where the $\kvec$ on the rhs is considered as a vector field.

Ideally, we would like for $\mbf{Q} \doteq i \cx$ to be resemble a position operator. One of the necessary requirements for this correspondence is that $\mbf{Q}$ be Hermitian, that is,
\begin{equation}
    \secip{Q_a\psi}{\zeta} = \secip{\psi}{Q_a\zeta}.
\end{equation}
We will call $\cx$ an anti-Hermitian connection if this holds. We will also be interested in the cases when $\cx$ possesses rotational symmetry:
\begin{definition}
    The connection $\cx$ is rotationally symmetric if
    \begin{equation}\label{eq:poincare_symm}
        \cx_{R \vf{X}} = \Sigma_R\cx_{\vf{X}} \Sigma_{R^{-1}}
    \end{equation}
    for any $R \in \SO(3)$, or equivalently (\cite{Hall2015}, Prop. C.7), if $\cx$ is a vector operator:
    \begin{equation}
        [\cx_a,\svop{J}_b] = i\epsilon_{abc}\cx_{c}.
    \end{equation}
\end{definition}

We will later discuss particular choices of a connection on $E$. For now, assume we have chosen a connection $\cx$. We then define the splitting of $\svop{J}$ induced by $\cx$:
\begin{align}
    \svop{L}^{\cx} &\doteq -i\svop{P} \times \cx = -i\kvec \times \cx \label{eq:LD}\\
    \svop{S}^{\cx}&\doteq \svop{J} - \svop{L}^{\cx} \label{eq:SD}
\end{align}
where $\kvec$ is considered as a vector field on $M$.
We will explore the extent to which $\SDbold$ and $\LDbold$ can be considered SAM and OAM operators. In particular, we will show how the properties of $\cx$ determine if they are vector operators and satisfy angular momentum commutation relations.

\begin{theorem}
    Suppose $\cx$ is rotationally symmetric. Then $\LDbold$ and $\SDbold$ are vector operators:
    \begin{align}
        [\LD_m,\sop{J}_n] &= i\epsilon_{mnp}\LD_p \\
        [\SD_m,\sop{J}_n] &= i\epsilon_{mnp}\SD_p
    \end{align}
    where the components are defined with respect to some right handed orthonormal basis of $(\mbf{v}_1,\mbf{v}_2,\mbf{v_3})$ of $\mathbb{R}^3$, that is, $\LD_m \doteq \mbf{v}_m \cdot \LDbold$ and $\SD_m \doteq \mbf{v}_m \cdot \SDbold$.
\end{theorem}
\begin{proof}
    Let $\mbf{v}$ be a constant vector field on $M$. Then by rotational symmetry
    \begin{align}
        \cx_{R\mbf{v}\times \kvec} = \tilde{\Sigma}_R \cx_{\mbf{v}\times R^{-1}\kvec}\tilde{\Sigma}_{R^{-1}} =
        \tilde{\Sigma}_R \cx_{\mbf{v}\times \kvec}\tilde{\Sigma}_{R^{-1}}
    \end{align}
    where we have used the fact that the vector field $\kvec$ is rotationally invariant. Thus by Eq. (\ref{eq:cross_product}):
    \begin{align}
        &R\mbf{v}\cdot (\kvec\times \cx)
        =\tilde{\Sigma}_R \mbf{v}\cdot (\kvec\times \cx)\tilde{\Sigma}_{R^{-1}} \\
        &(R\mbf{v})\cdot \LDbold = \tilde{\Sigma}_R (\mbf{v}\cdot \LDbold)\tilde{\Sigma}_{R^{-1}}
    \end{align}
    Let $R_{n}(t)$ denote the rotation by angle $t$ about $v_n$. Then
    \begin{gather}
        \frac{d}{dt}\Big|_{t=0} (R_{n}(t)\mbf{v}_m)\cdot \LDbold = \frac{d}{dt}\Big|_{t=0}\tilde{\Sigma}_{R_{n}(t)}(\mbf{v}_m\cdot \LDbold)\tilde{\Sigma}_{R_{n}(-t)} \\
        (\mbf{v}_n \times \mbf{v}_m) \cdot \LDbold = \frac{d}{dt}\Big|_{t=0}e^{-i(\mbf{v}_n \cdot \svop{J})t} \LD_m  e^{i(\mbf{v}_n \cdot \svop{J})t} \\
        -\epsilon_{mnp}\LD_p = -i(\sop{J}_n\LD_m - \LD_m\sop{J}_n) \\
        i\epsilon_{mnp} \LD_p = [\LD_m,\sop{J}_n]
    \end{gather}
    proving that $\LDbold$ is a vector operator. Since $\svop{J}$ and $\LDbold$ are vector operators, $\SDbold$ is also a vector operator.
\end{proof}
We note that the rotational invariance of $\cx$ is a very mild condition and is expected of any useful connection; it essentially imposes the condition that $\cx$ be constructed from vector operators. We now turn to the question of whether $\SDbold$ is an internal operator as one would expect for an SAM operator. An operator $\bar{\mbf{S}}$ on a vector bundle $\pi:E\rightarrow M$ is internal if it preserves fibers, that is, if it only changes the polarization but not the momentum. The corresponding operator $\svop{S}$ on $L^2(E)$ is thus internal if for any $\psi$ and $\kvec_0$, $[\svop{S}\psi](\kvec_0)$ depends only on $\psi(\kvec_0)$ and not on $\psi$ at any other $\kvec$. This is equivalent to the condition that $\svop{S}$ is a point operator on $L^2(E)$ \cite{Tu2017differential}.
\begin{theorem}
    $\SDbold$ is a an internal operator for any connection $\cx$.
\end{theorem}
\begin{proof}
    By Lemma 7.23 in Ref. \cite{Tu2017differential}, to show $\SDbold$ is internal it is sufficient to show that for any $\psi \in L^2(E)$ and any smooth function $f:M\rightarrow \Comp$,
    \begin{equation}
        \SDbold(f\psi) = f\SDbold(\psi),
    \end{equation}
    and for this it is sufficient to prove each component of $\SDbold$ is $f$-linear.
    We have
    \begin{align}
        [\sop{J}_m(f\psi)](\kvec) &= i\frac{d}{dt}\Big|_{t=0} [\tilde{\Sigma}_{R_m(t)}(f\psi)](\kvec) \\
        &= i\frac{d}{dt}\Big|_{t=0}\Sigma_{R_m(t)}\big[ f\big(R_m(-t)\kvec\big)\psi\big(R_m(-t)\kvec\big)\big] \\
        &= if(\kvec)[\tilde{\Sigma}_{R_m(t)}\psi](\kvec) + i\psi(\kvec)\frac{d}{dt}\Big|_{t=0}f\big(R_m(-t)\kvec\big) \\
        &= f(\kvec)[\sop{J}_m\psi](\kvec) + i \psi(\kvec)\frac{d}{dt}\Big|_{t=0}f\big(\kvec -\mbf{v}_m \times \kvec t + O(t^2)\big)
    \end{align}
    so that
    \begin{equation}
        \sop{J}_m(f\psi) = f\sop{J}_{m}\psi - i \psi(\kvec)df(\mbf{v}_m\times \kvec).
    \end{equation}
    By the Leibniz rule for $\cx$,
    \begin{align}
        \LD_m(f\psi)&= -i(\kvec\times \cx)_m (f\psi) \\
        &= -i\cx_{\mbf{v}_m \times \kvec}(f\psi) \\
        &= f\LD_m\psi - i\psi df(\mbf{v}_m\times \kvec).
    \end{align}
    Thus,
    \begin{align}
        \SD_m(f\psi) &= \sop{J}_m(f\psi)-\LD_m(f\psi)\\
        &=f(\sop{J}_m\psi -\LD_m \psi) \\
        &= f\SD_m(\psi),
    \end{align}
    proving the result.
\end{proof}
We now address the fundamental question of whether or not $\SDbold$ and $\LDbold$ satisfy $\so(3)$ commutation relations so that they are valid angular momentum operators. In contrast to the previous two results, this imposes a very stringent condition on $\cx$---one which we will show cannot be satisfied for massless particles. As for the affine connections encountered in general relativity, the curvature of a connection $\cx$ is defined as the map \cite{Tu2017differential}
\begin{gather}
    F_{\cx}:\mathfrak{X}(M) \times \mathfrak{X}(M) \times L^2(E)\rightarrow L^2(E) \nonumber \\
    F_{\cx}(\vf{X},\vf{Y})\psi = (\cx_{\vf{X}}\cx_{\vf{Y}} - \cx_{\vf{X}} \cx_{\vf{Y}} - \cx_{[\vf{X},\vf{Y}]})\psi \label{eq:curvature}
\end{gather}
where $[\vf{X},\vf{Y}]$ is the Jacobi-Lie bracket of vector fields. The curvature is most intuitively understood in terms of holonomy. That is, one can think of $F_{\cx}(\vf{X},\vf{Y})$ as measuring the change in a vector which is parallel transported via $\cx$ along an infinitesimal loop in $M$ in the plane specified by $\vf{X}$ and $\vf{Y}$. We note here the important fact that $F_{\cx}$ is defined pointwise in all three arguments (\cite{Tu2017differential}, Sec. 10.3), so $F_{\cx}(\mbf{u_1},\mbf{u_2})v$ is well-defined for $\mbf{u}_1,\mbf{u}_2 \in T_{\kvec}M$ and $v \in E(\kvec)$. The same is not true of $\cx$ which is defined pointwise in $\mathfrak{X}(M)$ but only locally in $L^2(E)$, that is, to define $[D_{\mbf{u}}\psi](\kvec)$, the section $\psi$ must be specified on a neighborhood of $\kvec$. The connection $\cx$ is called flat if $F_{\cx}$ vanishes identically. Let $S^2_r \seq \mathbb{R}^3$ denote the sphere of radius $r$ in $\kvec$ space, and $E|_{S^2_r}$ be the bundle obtained by restricting the base manifold of $E$ to $S^2_r$. Physically, this corresponds to considering monochromatic waves with energy $\sqrt{m^2 + r^2}$. We then define a slightly weaker notion of flatness of a connection:
\begin{definition}[$S^2$-flat connections]
    A connection $\cx$ on $E$ is $S^2$-flat if it is flat on $E|_{S^2_r}$ for every $r>0$. Equivalently, $\cx$ is $S^2$-flat if $F_{\cx}(\vf{X},\vf{Y}) = 0$ for all vector fields transverse to $\kvec$:
    \begin{subequations}\label{eq:vf_transversality}
    \begin{align}
        \vf{X}(\kvec)\cdot \khat &= 0  \\
        \vf{Y}(\kvec)\cdot \khat &= 0.
    \end{align}
    \end{subequations}
\end{definition}
Remarkably, the condition that $\SDbold$ and $\LDbold$ satisfy angular momentum relations is equivalent to the condition that the connection $\cx$ is $S^2$-flat.
\begin{theorem}\label{thm:curvature_theorem}
    $\SDbold$ and $\LDbold$ satisfy the angular momentum commutation relations
    \begin{align}
        [\SD_a, \SD_b] = i\epsilon_{abc}\SD_c \label{eq:SD_comm}\\
        [\LD_a, \LD_b] = i\epsilon_{abc}\LD_c \label{eq:LD_comm}
    \end{align}
    if and only if the connection $\cx$ is $S^2$-flat.
\end{theorem}
\begin{proof}
    We first note that since $[\sop{J}_a,\sop{J}_b]=\sop{J}_c$, Eqs. (\ref{eq:SD_comm}) and (\ref{eq:LD_comm}) are equivalent, so it suffices to consider only the commutation relations for $\LDbold$. In general, we have
    \begin{align}
        [\LD_a,\LD_b] &= -(\cx_{\uv{e}_a\times \kvec} \cx_{\uv{e}_b\times \kvec} - \cx_{\uv{e}_b\times \kvec} \cx_{\uv{e}_a\times \kvec}) \\
        &= -\cx_{[\uv{e}_a\times\kvec,\uv{e}_b\times\kvec]} - F_{\cx}(\uv{e}_a \times \kvec, \uv{e}_b \times \kvec).
    \end{align}
    By direct computation $[\uv{e}_a\times\kvec,\uv{e}_b\times\kvec]= -\epsilon_{abc}(\uv{e}_c \times \kvec)$, so
    \begin{align}
        [\LD_a,\LD_b] &= \epsilon_{abc}\cx_{\uv{e}_c \times \kvec} - F_{\cx}(\uv{e}_a \times \kvec, \uv{e}_b \times \kvec) \\
        &= i\epsilon_{abc} \LD_c - F_{\cx}(\uv{e}_a \times \kvec, \uv{e}_b \times \kvec). \label{eq:curvature_comm}
    \end{align}
    We thus see that $\LDbold$ satisfies angular momentum commutation relations if and only if
    \begin{equation}
        F_{\cx}(\uv{e}_a \times \kvec, \uv{e}_b \times \kvec) = 0
    \end{equation}
    for each $a,b$. Since $F_{\cx}$ is defined pointwise in all arguments, this condition is equivalent to
    \begin{equation}
        F_{\cx}(\vf{X} \times \kvec,\vf{Y}\times \kvec) = 0
    \end{equation}
    for any vector fields $\vf{X}, \vf{Y}$, and this in turn is equivalent to
    \begin{equation}
        F_{\cx}(\vf{X}^\perp,\vf{Y}^\perp) = 0
    \end{equation}
    for all vector fields $\vf{X}^\perp$ and $\vf{Y}^\perp$ satisfying the $S^2$-transversality condition Eq. (\ref{eq:vf_transversality}). Thus, the angular momentum commutation relations Eqs. (\ref{eq:SD_comm}) and (\ref{eq:LD_comm}) are satisfied if and only if $\cx$ is $S^2$-flat.
\end{proof}

\section{Applications}\label{sec:Applications}

\subsection{Massless particles do not admit an SAM-OAM decomposition}
By Theorem \ref{thm:curvature_theorem} a valid SAM-OAM decomposition of the form (\ref{eq:LD})-(\ref{eq:SD}) requires the existence of an $S^2$-flat connection. No such connection exists for massless particles with helicity $h\neq 0$.
\begin{theorem}[Massless SAM-OAM no-go theorem]\label{thm:nogo_thm}
    The massless particle bundle $\gamma_h$ with $h\neq 0$ admits no SAM-OAM decomposition of the form $\svop{J} = \svop{L}^{\cx} + \svop{S}^{\cx}$ where $\LDbold = -i\kvec\times \cx$ and $\cx$ is a connection.
\end{theorem}
\begin{proof}
The massless particle bundle of a helicity $h$ particle is isomorphic to $\pi:\gamma_h \rightarrow \pspace$ and has (first) Chern number $C(\gamma_h)=-2h$ \cite{PalmerducaQin_helicity}. When calculating the Chern number, we can restrict the base manifold to the unit sphere $S^2$, that is, $C(\gamma_h) = C(\gamma_h|_{S^2})$ (\cite{PalmerducaQin_GT}, Eq. 4.1). The Chern number can be calculated in terms of any connection $\cx$ on $\gamma_h$, even though the result is independent of the choice of connection \cite{Tu2017differential}. In particular, smoothly choose unit vectors $\uv{e}_h(\khat) \in \gamma_h|_{S^2}$ for $\khat \in S^2\setminus p$. We have removed a single point $p$ so that a smooth choice is possible. Then we can express the curvature in this basis:
\begin{equation}
    F_{\cx}(\vf{X}^\perp, \vf{Y}^\perp)\mbf{e}_h = \Omega(\vf{X}^\perp, \vf{Y}^\perp)\mbf{e}_h
\end{equation}
where $\Omega$ is a differential 2-form on $S^2\setminus p$, called the connection form of $\cx$ with respect to $\uv{e}_h$ \cite{Tu2017differential}. Then \cite{Tu2017differential}
\begin{equation}
    C(\gamma_h) = C(\gamma_h|_{S^2}) = \int_{S^2} \Omega = -2h \neq 0.
\end{equation}
Since the result is nonzero, $\Omega$ must be nonzero for some points on $S^2$. Therefore, there exist $\vf{X}^\perp$ and $\vf{Y}^\perp$ such that $F_{\cx}(\vf{X}^\perp, \vf{Y}^\perp) \neq 0$, and thus $\cx$ is not $S^2$-flat. By Theorem \ref{thm:curvature_theorem}, $\cx$ does not produce an SAM-OAM decomposition. $\cx$ is a completely arbitrary connection on $\gamma_h$, so this proves the theorem.
\end{proof}

This result follows a recent series of no-go theorems we recently proved pertaining to massless SAM-OAM splittings \cite{PalmerducaQin_SAMOAM}. Strictly speaking, Theorem \ref{thm:curvature_theorem} is a special case of No-Go Theorem 1 in Ref. \cite{PalmerducaQin_helicity}, which showed that no massless SAM-OAM decomposition is possible without constraints on the form of SAM and OAM operators. However, the proofs proceed by different arguments, with that in Ref. \cite{PalmerducaQin_SAMOAM} based on the classification of vector space UIRs of $\SO(3)$, whereas here the argument is based on the non-vanishing curvature of connections.

\subsection{The boost connection and the \texorpdfstring{$\mbf{J}_\parallel$, $\mbf{J}_\perp$}{massless} splitting}
We now show that the massless splitting \cite{Bliokh2010,Bialynicki-Birula2011} discussed in the introduction
\begin{gather}
    \svop{J} = \svop{J}_\perp + \svop{J}_\parallel \\
    \svop{J}_\perp = -\frac{1}{\sop{H}}\svop{P} \times \svop{K} = \frac{1}{\kmag} \kvec \times \svop{K} \\
    \svop{J}_\parallel = (\phats \cdot\svop{J})\phats
\end{gather}
is induced by a connection, but that the connection has nonvanishing curvature, explaining why $\svop{J}_\parallel$ and $\svop{J}_\perp$ do not satisfy angular momentum commutation relations. 
\begin{theorem}[Boost connection]
    Suppose $E$ is a particle bundle with mass $m \geq 0$. Then
    \begin{align}
        \DK &= -\frac{i}{2}\Big(\frac{1}{\sop{H}}\svop{K} + \svop{K}\frac{1}{\sop{H}}\Big) \\
        &= - \frac{i}{\sop{H}}\Big(\svop{K} - \frac{i\svop{P}}{2\sop{H}}\Big) \label{eq:DK_b}
    \end{align}
    is an anti-Hermitian, rotationally symmetric connection on $E$ which we call the boost connection. For massless particles, this connection induces the separation of $\svop{J}$ into $\svop{L}^{K} = \svop{J}_\perp$ and $\svop{S}^{K} = \svop{J}_\parallel$ (which are vector operators but not angular momentum operators).
\end{theorem}
\begin{proof}
    We will first show that 
    \begin{equation}\label{eq:DK_bar}
        \barDK = -\frac{i}{\sop{H}}\svop{K}
    \end{equation}
    is a connection. For $\vf{X} \in \mathfrak{X}(M)$, $\psi \in L^2(E)$, and smooth $f:M\rightarrow \Comp$, we have
    \begin{equation}
        \barDK_{f\vf{X}}\psi = -i \frac{f}{k^0}(\vf{X} \cdot \svop{K})\psi
    \end{equation}
    so $\DK$ is $f$-linear in $X$. We now check that the Leibniz rule is satisfied in $\psi$. Let $\Lambda_{\mbf{v}}$ denote a boost by $\mbf{v} \in \Real^3$. Then 
    \begin{align}
        [\barDK_{\vf{X}}(f\psi)](\kvec) &=-\frac{i}{k^0}(\vf{X}\cdot \svop{K})[(f\psi)](\kvec) \\
        &= \frac{1}{k^0}\frac{d}{dt}\Big|_{t=0} \Big[\tilde{\Sigma}(\Lambda_{-\vf{X}(\kvec)t})(f\psi)\Big](\kvec) \\
        &= \frac{1}{k^0}\frac{d}{dt}\Big|_{t=0} \Sigma(\Lambda_{-\vf{X}(\kvec)t})f\big(\Lambda_{\vf{X}(\kvec)t}\kvec\big)\psi(\Lambda_{\vf{X}(\kvec)t}) \\
        &= f(\kvec)\DK_{\vf{X}}\psi(\kvec) + \frac{\psi(\kvec)}{k^0} \frac{d}{dt}\Big|_{t=0}f\big(\Lambda_{\vf{X}(\kvec)t}\kvec\big).
    \end{align}
    Note that
    \begin{equation}
        (\Lambda_{\vf{X}(\kvec)t}\kvec)^m = \tensor{(\Lambda_{\vf{X}(\kvec)t})}{^m_\mu} k^\mu = k^m + k^0X^m(\kvec)t + O(t^2)
    \end{equation}
    so
    \begin{align}
        \barDK_{\vf{X}}(f\psi) = f\barDK_{\vf{X}}\psi+\psi(\kvec)df(\vf{X})
    \end{align}
    proving that $\barDK$ satisfies Eq. (\ref{eq:Leibniz}) and is thus a connection. Now note that $\frac{\svop{P}}{\sop{H}^2}$ is a point operator, \ie, for any functions $f,g$,
    \begin{equation}
        \Big[(f\vf{X})\cdot\frac{\svop{P}}{\sop{H}^2}\Big](g\psi) = fg(\vf{X}\cdot\frac{\svop{P}}{\sop{H}^2})\psi.
    \end{equation}
    It thus follows that
    \begin{equation}
        \DK = \barDK -i\frac{\svop{P}}{\sop{H}^2}
    \end{equation}
    also satisfies $f$-linearity in $\vf{X}$ and the Leibniz rule in $\psi$, and is thus a connection. Using Eq. (\ref{eq:comm_Ka_hn}), $\DK$ can be expressed in the manifestly anti-Hermitian form
    \begin{equation}
        \DK = -\frac{i}{2}\Big(\frac{1}{\sop{H}}\svop{K} + \svop{K}\frac{1}{\sop{H}}\Big).
    \end{equation}
    $\DK$ is rotationally invariant since $\sop{H}^{-1}\svop{K}$ and $\svop{K}\sop{H}^{-1}$ are vector operators. In the massless case, this connection induces the splitting of $\svop{J}$ into $\svop{L}^K$ and $\svop{S}^K$ according to Eqs. (\ref{eq:LD})-(\ref{eq:SD}), where
    \begin{align}
        \svop{L}^K &\doteq -i \svop{P} \times \DK \\
        &= -\frac{1}{\sop{H}}\svop{P} \times \svop{K}\\ 
        &= \svop{J}_\perp.
    \end{align}
    This implies too that $\svop{S}^K \doteq \svop{J} - \svop{L}^K =  \svop{J}_\parallel$. By Theorem \ref{thm:nogo_thm}, $\DK$ cannot be $S^2$-flat when $m=0$ and thus $\svop{L}^K$ and $\svop{S}^K$ are not angular momentum operators. It will be shown in Theorem \ref{thm:Boost_connection_not_flat} that this is true for all $m\geq 0$. 
\end{proof}
Intuitively, $\DK$ can be understood in terms of the corresponding parallel transport. If $s(t)$ is a path between $\kvec_0$ and $\kvec_1$, $\DK$ identifies the vectors in $E(\kvec_0)$ with those in $E(\kvec_1)$ by boosting along the path $s(t)$. This connection is defined for both massive and massless particles. Regardless of the particle's mass, this connection is not $S^2$-flat and thus the operators $\svop{L}^K$ and $\svop{S}^K$ induced via Eqs. (\ref{eq:LD}) and (\ref{eq:SD}) are not angular momentum operators. This is ultimately due to Wigner rotation \cite{Wigner1939}. Boosts in different directions do not commute ($[K_a,K_b] = -i\epsilon_{abc}J_c$) so boosting around a loop results in a rotation rather than an identity transformation, reflecting the non-zero curvature of this connection. 

We now show this more rigorously. The calculation of curvatures is rather delicate. We will calculate the curvature of the boost connection in two ways; the first method is simpler but does not easily generalize. The second method requires more work but allows for more complicated calculations which we will need later in this article.

Let $(\ek,\uv{e}_\theta,\uv{e}_\phi)$ be the standard spherical unit vectors
\begin{align}
    &\ek = \sin(\theta) \cos(\phi) \uv{e}_1 + \sin(\theta) \sin(\phi) \uv{e}_2 + \cos(\theta) \uv{e}_3\\
    &\uv{e}_\theta = \cos(\theta)\cos(\phi)\uv{e}_1 + \cos(\theta)\sin(\phi)\uv{e}_2 - \sin(\theta)\uv{e}_3 \\
    &\uv{e}_\phi = -\sin(\phi)\uv{e}_1 + \cos(\phi)\uv{e}_2.
\end{align}
It will be useful to use the radial cylindrical unit vector $\uv{e}_R$ and its relation to the spherical unit vectors:
\begin{equation}
    \uv{e}_R = \cos(\phi)\uv{e}_1 + \sin(\phi)\uv{e}_2 = \sin(\theta)\khat + \cos(\theta)\uv{e}_\theta.
\end{equation}
At $\kvec_0$, let $(\uv{u}_1,\uv{u}_2,\uv{u}_3 = \khat_0)$ be an orthonormal basis. We will extend $\uv{u}_a$ as constant vector fields on all of $M$, specifically, $\uv{u}_3(\kvec) \neq \ek$ for general $\kvec \neq \kvec_0$.
\begin{theorem}\label{thm:Boost_connection_not_flat}
    For particles of any mass, the boost connection $\DK$ is not $S^2$-flat, in particular, the $S^2$-curvature at $\kvec_0$ is given by
    \begin{equation}
        F_K(\uv{u}_1,\uv{u}_2) = \frac{i}{\sop{H}^2}\sop{J}_{3}.
    \end{equation}
    This can also be expressed as
    \begin{equation}\label{eq:boost_curvature}
        F_K(\etheta,\ephi) = \frac{i}{\sop{H}^2}\sop{J}_k.
    \end{equation}
    Since the connection is not $S^2$-flat, the operators
    \begin{align}
         \svop{L}^K &\doteq -i \svop{P} \times \DK \\
         \svop{S}^K &\doteq \svop{J} - \svop{L}^K
    \end{align}
    are not angular momentum operators.
\end{theorem}
\begin{proof}[Proof 1]
This first proof uses Cartesian components of the generators. We will calculate the curvature at some arbitrary but fixed $\kvec_0$. We orient our axes so that $\sop{K}_a = \uv{u}_a \cdot \svop{K}$ and similarly for other generators. The Cartesian components of $\DK$ are
\begin{equation}
    \DK_a \doteq \DK_{\uv{u}_a} = -i\Big(\frac{\sop{K}_a}{\sop{H}} - i\frac{\sop{P}_a}{\sop{H}^2}\Big).
\end{equation}
Since the $\uv{u}_a$ are constant vector fields, their Jacobi-Lie bracket vanishes: $[\uv{u}_a,\uv{u}_b] = 0$. Thus
\begin{align}
    F_K(\uv{u}_1,\uv{u}_2) &= [\DK_1,\DK_2] \\
    &= -\Big[ \frac{1}{\sop{H}}\sop{K}_1 - i\frac{1}{\sop{H}^2}\sop{P}_1 , \frac{1}{\sop{H}}\sop{K}_2 - i\frac{1}{\sop{H}^2}\sop{P}_2\Big].
\end{align}
We are interested only in calculating this quantity at $\kvec_0$. We note that while $F_K(\uv{u}_1,\uv{u}_2)$ and $P_a$ are defined pointwise, this is not true of the operators $\DK_a$, $K_a$, and $J_a$ which are locally defined. Thus, even though $\sop{P}_1 = 0$ at $\kvec_0$, $[P_1,J_2] = iP_3 \neq 0$, illustrating that we must evaluate the commutators before evaluating pointwise. By Eq. (\ref{eq:comm_Ka_hn}),
\begin{align}
    [\sop{H}^{-1}\sop{K}_1,\sop{H}^{-1}\sop{K}_2] &= \sop{H}^{-2}[\sop{K}_1,\sop{K}_2] + \sop{H}^{-1}[\sop{K}_1,\sop{H}^{-1}]\sop{K}_2 + \sop{H}^{-1}[\sop{H}^{-1},\sop{K}_2]\sop{K}_1 \\
    &= -i\sop{H}^{-2}\sop{J}_3 + i\sop{H}^{-3}(-\sop{P}_1\sop{K}_2 + \sop{P}_2\sop{K}_1).
\end{align}
By Eq. (\ref{eq:comm_id_hm_k_hn_p}),
\begin{align}
    [\sop{H}^{-1}\sop{K}_1,\sop{H}^{-2}\sop{P}_2]  &= -2i\sop{H}^{-4}\sop{P}_1\sop{P}_2 \\
    [\sop{H}^{-2}\sop{P}_1,\sop{H}^{-1}\sop{K}_2]  &= i\sop{H}^{-4}\sop{P}_1\sop{P}_2
\end{align}
Thus, evaluating $F_K(\uv{u}_1,\uv{u}_2)$ at $\kvec_0$ where $\sop{P}_1=\sop{P}_2 = 0$, we have
\begin{alignat}{3}
    F_K(\uv{u}_1,\uv{u}_2) &= &&-[\sop{H}^{-1}\sop{K}_1, \sop{H}^{-1}\sop{K}_2]\\
    & &&+i[\sop{H}^{-1}\sop{K}_1,\sop{H}^{-2}\sop{P}_2] + i[\sop{H}^{-2}\sop{P}_1,\sop{H}^{-1}\sop{K}_2] \\
    &=&&i\sop{H}^{-2}\sop{J}_3.
\end{alignat}
This result holds for any orthonormal basis with $\uv{u}_3 = \kvec_0$, so we can choose $\uv{u}_1 = \etheta(\kvec_0)$ and $\uv{u}_2 = \ephi(\kvec_0)$. By varying $\kvec_0$ we obtain 
\begin{equation}
    F_K(\etheta,\ephi) = i\sop{H}^{-2}J_k.
\end{equation}
\end{proof}
\begin{proof}[Proof 2]
    We can instead directly do the global calculation
    \begin{equation}\label{eq:FK_theta_phi_def}
    F_K(\etheta,\ephi) = [\DK_{\etheta},\DK_{\ephi}] - \DK_{[\etheta,\ephi]}.
    \end{equation}
    This requires more work but there are  critical advantages of this method which we will later exploit. One computation advantage is immediately apparent. If we let $\sop{P}_\theta = \etheta \cdot \svop{P}$, $\sop{P}_{\phi} = \ephi \cdot \svop{P}$, and similarly for the other generators, then we have
    \begin{equation}
        [\sop{P}_\theta\psi](\kvec) = (\etheta \cdot \kvec) \psi(\kvec) = 0.
    \end{equation}
    So globally $\sop{P}_\theta = 0$ and likewise $\sop{P}_\phi = 0$. Therefore,
    \begin{subequations}\label{eq:K_conn_spherical}
    \begin{align}
        \DK_{\etheta} &= -i\sop{H}^{-1}\sop{K}_\theta = -\frac{i}{\sop{H}}(\cos\theta \cos\phi \sop{K}_1 + \cos \theta \sin \phi \,\sop{K}_2 - \sin \theta \, \sop{K}_3)\\
        \DK_{\ephi} &= -i\sop{H}^{-1}\sop{K}_\phi = -\frac{i}{\sop{H}}(-\sin\phi \,\sop{K}_1 + \cos \phi \,\sop{K}_2).
    \end{align}
    \end{subequations}
    The most important advantage is that while the Cartesian boost generators $\sop{K}_a$ do not commute with $|\svop{P}|$ and $\sop{H}$, the spherical boosts $\sop{K}_{\theta,\phi}$ do. Indeed, using Eq. (\ref{eq:comm_Ka_P}) we have
    \begin{align}
        [\sop{K}_{\theta,\phi},|\svop{P}|] &= \uv{e}_{\theta,\phi} \cdot[\svop{K},|\svop{P}|] = i\frac{\sop{H}}{|\svop{P}|}\sop{P}_{\theta,\phi} = 0 \\
        [\sop{K}_{\theta,\phi},\sop{H}] &= \uv{e}_{\theta,\phi} \cdot [\svop{K},\sop{H}] = i\sop{P}_{\theta,\phi} = 0
    \end{align}

    The Jacobi-Lie bracket of $\etheta$ and $\ephi$ is
    \begin{equation}\label{eq:JL_theta_phi}
        [\etheta,\ephi] = -\frac{\cot{\theta}}{\kmag}\ephi.
    \end{equation}
    Thus,
    \begin{equation}\label{eq:DK_JL_theta_phi}
        \DK_{[\etheta,\ephi]} =  -\frac{\cot \theta}{\kmag}\DK_{\ephi} = \frac{i \cot \theta}{\kmag \sop{H}}\sop{K}_\phi.
    \end{equation}
    We now turn to the remaining term in Eq. (\ref{eq:FK_theta_phi_def}). Using Eqs. (\ref{eq:K_conn_spherical}) and the Leibniz rule, we obtain
    \begin{alignat}{2}
        \DK_\phi\DK_\theta\psi &= -\frac{i}{\sop{H}}&&\DK_\phi(\cos{\theta}\cos{\phi}\,\sop{K}_1\psi + \cos\theta \sin \phi \,\sop{K}_2\psi -\sin\theta\,\sop{K}_3\psi) \\
        &= -\frac{i}{h}&&\Big\{\cos\theta \cos\phi \, \DK_\phi(\sop{K}\psi) + \frac{\cos \theta}{\sin \theta\,\kmag}\frac{\partial}{\partial \phi}(\cos\phi)(\sop{K}_1\psi) \\
        & &&+\cos \theta \sin \phi \,\DK_\phi (\sop{K}_2\psi) + \frac{\cos \theta}{\sin\theta \,\kmag}\frac{\partial}{\partial \phi}(\sin \phi)(\sop{K}_2\psi) \nonumber\\
        & &&-\sin\theta \DK_\phi(\sop{K}_3 \psi) \Big\} \nonumber\\
        &=\frac{1}{\sop{H}^2}&&\Big(\cos\theta \cos \phi \sin \phi \,\sop{K}_1^2 - \cos \theta \cos^2\phi \,\sop{K}_2\sop{K}_1 \label{eq:DK_phi_theta} \\
        & &&+ \cos \theta \sin^2 \phi \, \sop{K}_1\sop{K}_2 - \cos\theta \cos \phi \sin \phi \, \sop{K}_2^2  \nonumber\\
        & &&-\sin \theta \sin \phi \, \sop{K}_1 \sop{K}_3 + \sin \theta \cos \phi \, \sop{K}_2 \sop{K}_3\Big)\psi + \frac{\cot \theta}{\kmag}\DK_{\ephi}\psi \nonumber 
    \end{alignat}
    By the same method we find

    \begin{alignat}{1}
        \DK_\theta \DK_\phi = \frac{1}{\sop{H}^2}&\Big(\cos\theta \cos \phi \sin\phi\,\sop{K}_1^2 + \cos\theta \sin^2 \phi \sop{K}_2\sop{K}_1 - \sin\theta \sin\phi \,\sop{K}_3\sop{K}_1 \label{eq:DK_theta_phi}\\
        &-\cos \theta \cos^2 \phi \, \sop{K}_1 \sop{K}_2 - \cos \theta \cos \phi \sin \phi \sop{K}_2^2 + \sin \theta \cos \phi \sop{K}_3 \sop{K}_2\Big). \nonumber
    \end{alignat}

    From Eqs. (\ref{eq:DK_phi_theta}) and (\ref{eq:DK_theta_phi}) we obtain
    \begin{alignat}{2}
        [\DK_\theta,\DK_\phi] &= \frac{1}{\sop{H}^2}\Big(&&-\cos\theta\sin^2\phi\,[\sop{K}_1,\sop{K}_2] - \cos\theta \cos^2\phi\,[K_1,K_2] \nonumber\\
        & &&-\sin\theta \sin\phi \, [\sop{K}_3,\sop{K}_1] + \sin\theta \cos \phi [\sop{K}_3,\sop{K}_2]\Big) - \frac{\cot\theta}{\kmag}\DK_\phi \nonumber \\ 
        &=\frac{i}{\sop{H}^2}(&&\cos\theta \, \sop{J}_3 + \sin\theta \sin \phi\, \sop{J}_2 + \sin\theta \cos \phi \sop{J}_1) - \frac{\cot\theta}{\kmag}\DK_\phi \nonumber \\
        &=\frac{i}{\sop{H}^2}&&\sop{J}_k- \frac{\cot\theta}{\kmag}\DK_\phi. \label{eq:comm_DK_theta_DK_phi}
    \end{alignat}
    Then from Eqs. (\ref{eq:FK_theta_phi_def}), (\ref{eq:DK_JL_theta_phi}),and (\ref{eq:comm_DK_theta_DK_phi}) we obtain
    \begin{equation}
        F_K(\etheta,\ephi) = \frac{i}{\sop{H}^2}\sop{J}_k.
    \end{equation}
\end{proof}

\subsection{The rotation connection for elementary particles}
We have seen that any elementary particle bundle is endowed with a natural anti-Hermitian connection $\DK$ resulting from the \Poincare action, particularly from the boost transformations. However, there is another such connection induced by the rotations. Indeed, for a path along a spherical (fixed energy) surface in $\kvec$ space, one can transport vectors along the path by applying infinitesimal rotations rather than infinitesimal boosts. If the path involves changes in energy, \ie, radial changes in $k$-space, these are still accomplished by boosts. This gives the intuitive description of the following connection.
\begin{theorem}[Rotation connection]
    Suppose $E$ is a particle bundle with mass $m \geq 0$. Let $\tilde{E}$ be the restriction of this bundle to the base manifold $\pspace$; note that $\tilde{E} = E$ for massless particles. Then
    \begin{align}
        \DR &= -i
        \Bigg(\frac{\phats \times \svop{J}}{|\svop{P}|} +\frac{1}{\sop{H}}\phats(\phats\cdot\svop{K}) - \frac{i\svop{P}}{2\sop{H}^2}\Bigg) \label{eq:DR_1}\\
        &= -i
        \Bigg(\frac{\phats \times \svop{J}}{|\svop{P}|} - \frac{i\phats}{|\svop{P}|}\Bigg)-\frac{i}{2}\Bigg(\frac{1}{\sop{H}}\phats(\phats\cdot\svop{K}) +(\svop{K}\cdot \phats)\phats\frac{1}{\sop{H}}\Bigg) \label{eq:DR_2}
    \end{align}
    is an anti-Hermitian connection on $\tilde{E}$, which we call the rotation connection.
\end{theorem}
\begin{proof}
    We first show that
    \begin{equation}\label{eq:DR_bar}
        \barDR \doteq -i\Bigg(\frac{\phats \times \svop{J}}{|\svop{P}|} +\frac{1}{\sop{H}}\phats(\phats\cdot\svop{K})\Bigg)
    \end{equation}
    is a connection. It is trivial that $\barDR_{f\vf{X}}\psi = f\bar{\DR}_{\vf{X}} \psi$. It remains to check that the Leibniz rule (\ref{eq:Leibniz}) holds . It suffices to check this when $\vf{X}$ is everywhere orthogonal or parallel to $\uv{e}_k$. In the former case, we take $\vf{X}^\perp$ to an arbitrary vector field such that $\vf{X}^\perp(\kvec) \cdot \uv{e}_k = 0$ for all $\kvec$. Then
    \begin{equation}
        \Big[\vf{X}^\perp\cdot\big[\frac{1}{\sop{H}}\phats(\phats\cdot\svop{K})\big]\psi \Big](\vf{k}) = [\sop{H}^{-1}(\vf{X}^\perp \cdot\ek)(\ek \cdot \svop{K})]\psi = 0
    \end{equation}
    so
    \begin{align}
        [\barDR_{\vf{X}^\perp}(f\psi)](\kvec) &= -\frac{i}{\kmag}\big(\vf{X}^\perp\cdot(\ek \times \svop{J})\big)[f\psi](\kvec) \\
        &=-\frac{i}{\kmag}\big((\vf{X}^\perp \times \ek)\cdot \svop{J} \big)[f\psi](\kvec) \\
        &=\frac{1}{\kmag} \Bigg[\frac{d}{dt}\Big|_{t=0}e^{-i(\vf{X}^\perp\times\ek)\cdot \svop{J}t}(f\psi)\Bigg](\kvec) \\ 
        &=\frac{1}{\kmag} \frac{d}{dt}\Big|_{t=0}\tilde{\Sigma}\big(R_{\vf{X}^\perp\times \ek}(t)\big)[f\psi](\kvec) \\ 
        &=\frac{1}{\kmag} \frac{d}{dt}\Big|_{t=0}f\big(R_{\vf{X}^\perp\times \ek}(-t)\kvec\big)\Sigma\big(R_{\vf{X}^\perp\times \ek}(t)\big)\psi\big(R_{\vf{X}^\perp\times \ek}(-t)\kvec\big) \\ 
        &= f(\kvec)\DR_{\vf{X}^\perp}\psi(\kvec) + \frac{\psi(\kvec)}{\kmag}\frac{d}{dt}\Big|_{t=0}f\big(\kvec - t(\vf{X}^\perp\times \ek) \times \kvec + O(t^2) \big) \\ 
        &= f(\kvec)\DR_{\vf{X}^\perp}\psi(\kvec) + \frac{\psi(\kvec)}{\kmag}\frac{d}{dt}\Big|_{t=0}f\big(\kvec + t\kmag\vf{X}^\perp \big) \\ 
        &= f(\kvec)\DR_{\vf{X}^\perp}\psi(\kvec) + \psi(\kvec) df(\vf{X}^\perp)(\kvec).
    \end{align}
    This shows that the Leibniz rule holds when $\vf{X} = \vf{X}^\perp$. We now consider the parallel case, where $\vf{X} = \ek$. In this case, dotting $\ek$ into Eq. (\ref{eq:DR_bar}) kills the first term. Then using the fact that $\barDK$ defined in Eq. (\ref{eq:DK_bar}) is a connection, we have
    \begin{align}
        [\barDR_{\ek}(f\psi)](\kvec) &= -\frac{i}{\sop{H}}(\ek\cdot \svop{K})[f\psi] =\barDK_{\ek}(f\psi) \\
        & = f\barDK_{\ek}\psi(\kvec) + \psi(\kvec)df(\ek) \\
        &=  f\barDR_{\ek}\psi(\kvec) + \psi(\kvec)df(\ek),
    \end{align}
    showing that $\barDR$ is a connection. Since 
    \begin{equation}
    -\frac{i\svop{P}}{2\sop{H}^2}(f\psi) = -f\frac{i\svop{P}}{2\sop{H}^2}(\psi)
    \end{equation}
    it follows that
    \begin{equation}
        \DK = \barDK -\frac{i\svop{P}}{2\sop{H}^2}
    \end{equation}
    also satisfies the Leibniz rule and is thus a connection. By application of the \Poincare commutation relations, we have
    \begin{equation}
        \frac{1}{\sop{H}}\phats (\phats \cdot \svop{K}) = (\svop{K}\cdot \phats)\phats \frac{1}{\sop{H}} + \frac{i\svop{P}}{\sop{H}^2} - \frac{2i\phats}{|\svop{P}|}
    \end{equation}
    from which the equivalence of Eqs. (\ref{eq:DR_1}) and (\ref{eq:DR_2}) follows.

    $\DR$ is an anti-Hermitian connection if
    \begin{equation}\label{eq:QK_proof}
        \QR \doteq i\DR = \frac{1}{|\svop{P}|^2}
        \Bigg(\svop{P} \times \svop{J} - i\svop{P}\Bigg)+\frac{1}{2}\Bigg(\frac{1}{\sop{H}}\phats(\phats\cdot\svop{K}) +(\svop{K}\cdot \phats)\phats\frac{1}{\sop{H}}\Bigg)
    \end{equation}
    is Hermitian. The second term in parenthesis is clearly Hermitian. We have that
    \begin{align}
        (\svop{P} \times \svop{J})^\dagger_a &= \epsilon_{abc}(\sop{P}_b\sop{J}_c)^\dagger = \epsilon_{abc}\sop{J}_c\sop{P}_b \\
        &= \epsilon_{abc}(\sop{P}_b\sop{J}_c - [\sop{P}_b,\sop{J}_c]) \\
        &= \epsilon_{abc}(\sop{P}_b\sop{J}_c - i\epsilon_{bcd}\sop{P}_d) \\
        &= \epsilon_{abc}\sop{P}_a\sop{J}_b - 2i\sop{P}_a
    \end{align}
    so
    \begin{equation}
        (\phats \times \svop{J})^\dagger = (\phats \times \svop{J}) - 2i\phats
    \end{equation}
    so the first term in Eq. (\ref{eq:QK_proof}) is also Hermitian. Therefore $\QK$ is Hermitian.
\end{proof}
The rotation connection induces another splitting of the angular momentum via Eqs. (\ref{eq:LD}) and (\ref{eq:SD}) with
\begin{align}\label{eq:LR}
    &\svop{L}^R = -i\phats \times (\phats\times\svop{J}).
\end{align}
Note that while $\DR$ has a $|\svop{P}|^{-1}$ singularity at the origin (for massive particles), $\svop{L}^R$ does not suffer from a diverging singularity. However, $\phats$ is still discontinuous at $\kvec = 0$, so even if $\psi$ is smooth, $\svop{L}^R\psi$ will typically be discontinuous at the origin. This does not cause fundamental mathematical difficulties since $\psi$ is only defined almost-everywhere in the $L^2$ theory, but it is still unusual from a physical perspective and one does not expect such singularities in physically meaningful operators. It is thus unsurprising that we will find $\DR$ is also not $S^2$-flat and thus $\svop{L}^R$ and $\svop{S}^R$ are not angular momentum operators.
\begin{theorem}
    The rotation connection is not $S^2$-flat, in particular,
    \begin{equation}\label{eq:rotation_curvature}
        F^R(\etheta,\ephi) = \frac{i\sop{J}_k}{|\svop{P}|^2}.
    \end{equation}
    Thus, this connection does not induce an SAM-OAM splitting of the total angular momentum by Theorem \ref{thm:curvature_theorem}.
\end{theorem}
\begin{proof}
    Using the second proof technique used for Theorem \ref{thm:Boost_connection_not_flat}, we find
    \begin{equation}\label{eq:comm_DR_theta_DR_phi}
        [\DR_\theta,\DR_\phi] = \frac{i}{|\svop{P}|^2}(\sop{J}_k + \cot \theta\, \sop{J}_\theta).
    \end{equation}
    Also,
    \begin{align}
        (\DR_{[\etheta,\ephi]}\psi)(\kvec) &= -\Big[\frac{\cot\theta}{|\svop{P}|}\DR_{\ephi}\Big]\psi(\kvec) = \Big[\frac{i\cot\theta}{|\svop{P}|^2}\ephi\cdot(\ek \times\svop{J})\psi\Big](\kvec) \\
        &= \Big[\frac{i\cot\theta}{|\svop{P}|^2}\sop{J}_\theta\psi\Big](\kvec)
    \end{align}
    so
    \begin{equation}\label{eq:DR_JL_theta_phi}
        \DR_{[\etheta,\ephi]} = \frac{i\cot\theta}{|\svop{P}|^2}\sop{J}_\theta.
    \end{equation}
    We thus obtain the nonzero $S^2$-curvature
    \begin{equation}
        F^R(\etheta,\ephi) = [\DR_\theta,\DR_\phi] - \DR_{[\etheta,\ephi]} = \frac{i\sop{J}_k}{|\svop{P}|^2}.
    \end{equation}
\end{proof}

\subsection{The spin-orbital decomposition for massive particles and its singularity in the massless limit}
We have seen that boosts and rotations each induce connections on particle bundles of arbitrary mass. However, neither of these connections are $S^2$-flat and thus neither gives an SAM-OAM decomposition or a path independent way to identify the internal states at different momenta. We observe though that the $S^2$-curvatures of these two connections, given in Eqs. (\ref{eq:boost_curvature}) and (\ref{eq:rotation_curvature}), are very similar, being proportional to $\sop{J}_k$ and differing only by a factor of
\begin{equation}
    \frac{|\svop{P}|^2}{\sop{H}^2} = \frac{\kmag^2}{\kmag^2+ m^2}.
\end{equation}
It is thus natural to ask if these connections can be combined in such a way that the curvatures cancel, producing an $S^2$-flat connection. We will see that, for massive particles, the answer is yes. A general result from the theory of connections indicates how to proceed. While the sum of two connections is never a connection, any \emph{affine} sum of connections is. That is,
\begin{equation}
    \Df \doteq f\DK + (1-f)\DR
\end{equation}
is a connection for any smooth $f \in C^\infty(M)$ (\cite{Tu2017differential}, Prop. 10.5). As we wish to construct a rotationally symmetric connection, we will assume that $f = f(\kmag)$ depends only on $\kmag$. The $S^2$-curvature of $\Df$ is given by
\begin{alignat}{2}
    F^f(\etheta,\etheta) &= [&&f\DK_\theta + (1-f)\DR_{\theta}, f\DK_\phi + (1-f)\DR_{\phi}] \nonumber\\
    & &&- f\DK_{[\etheta,\ephi]} - (1-f)\DR_{[\etheta,\ephi]} \\
    &=&&f^2[\DK_\theta,\DK_\phi] + (1-f)^2[\DR_\theta,\DR_\phi] \nonumber\\
    & &&+ f(1-f)\Big([\DK_\theta,\DR_\phi] + [\DR_\theta,\DK_\phi]\Big) \nonumber\\
    & &&- f\DK_{[\etheta,\ephi]} - (1-f)\DR_{[\etheta,\ephi]} \label{eq:F_f}.
\end{alignat}
We note here that we have made use of the fact that the connection terms $\mbf{D}^{K,R}_{\theta,\phi}$ commute with $\kmag$ and thus with $f$. This is only true because we are working with the spherical components of the connection; had we worked with Cartesian components this would not be the case and this calculation would be substantially more difficult, if not intractable. The commutators $[\DK_\theta,\DK_\phi]$ and $[\DR_\theta,\DR_\phi]$ were calculated in Eqs. (\ref{eq:comm_DK_theta_DK_phi}) and (\ref{eq:comm_DR_theta_DR_phi}) and the Jacobi-Lie terms $\DK_{[\etheta,\ephi]}$ and $\DK_{[\etheta,\ephi]}$ in Eqs. (\ref{eq:DK_JL_theta_phi}) and (\ref{eq:DR_JL_theta_phi}). We can calculate the remaining commutators using the same method:
\begin{align}
    [\DK_\theta,\DR_\phi] &= \frac{i\sop{J}_k}{|\svop{P}|^2} + \frac{i\cot \theta}{\sop{H}|\svop{P}|}K_\phi \\
    [\DR_\theta,\DK_\phi] &= \frac{i\sop{J}_k}{|\svop{P}|^2} + \frac{i\cot \theta}{|\svop{P}|^2}K_\phi.
\end{align}
Plugging all of these expressions into Eq. (\ref{eq:F_f}) gives 
\begin{equation}
    F^f(\etheta,\ephi) = \Big(\frac{f^2}{\sop{H}^2} + \frac{1-f^2}{|\svop{P}^2|}  \Big)\sop{J}_k.
\end{equation}
Thus, if we require $\Df$ to be $S^2$-flat so as to produce an SAM-OAM decomposition of $\svop{J}$, then we find two solutions
\begin{equation}
    f_\pm(\kmag) = \pm \frac{\sop{H}}{m} = \pm \frac{\sqrt{\kmag^2 + m^2}}{m}
\end{equation}
and we denote the corresponding connections by $\Dpm$. Only $f_+$ produces a non-singular connection at the origin. Indeed, \emph{a priori} $\Df$ is only a legitimate connection if we remove the origin $\kvec = 0$ due to the $\kmag^{-1}$ singularity in $\DR$. However, for $\kmag\ll 1$
\begin{equation}
    (1-f_+) \sim -\frac{\kmag^2}{2m^2},
\end{equation}
so that for small $\kmag$ 
\begin{equation}
    (1-f_+)\DR \sim -\frac{i}{2m^2}\Big(\svop{P}\times \svop{J} + \frac{1}{\sop{H}}\svop{P}(\svop{P}\cdot \svop{K})-\frac{i|\svop{P}|^2}{2\sop{H}^2}\svop{P}\Big)
\end{equation}
which is non-singular at $\kvec = 0$. On the other hand
\begin{equation}
    \lim_{\kmag \rightarrow 0}(1-f_-) = 2,
\end{equation}
showing that $\Dm$ is singular at the origin. Thus, only $\Dp$ gives a globally well-defined flat connection for massive particle bundles.

Plugging in $f_+$ we find that
\begin{align}
    \Dp &= \frac{\sop{H}}{m}\DK + \Big(1-\frac{\sop{H}}{m}\Big)\DR \label{eq:NW_simple}\\
    &=-\frac{i}{\sop{H}}(\svop{K}-i\svop{P}/2\sop{H}) + \frac{i}{\sop{H}m(\sop{H}+m)}\svop{P}\times (\sop{H}\svop{J} + \svop{P}\times\svop{J}) \label{eq:NW_complicated}
\end{align}
In this latter form, we find that $\mbf{Q}^{NW} = i\Dp$ is precisely the Newton-Wigner position operator in Eq.  (\ref{eq:NW_Jordan}). We see that we can recover this operator via a novel method by studying connections on massive particle bundles which produce well-defined SAM-OAM splittings. Our method helps explain the rather complicated expression for $\mbf{Q}^{NW}$: it corresponds to an affine sum of the two natural connections induced by \Poincare symmetry, namely the rotation and boost connections. The affine parameter $\sop{H}/m$ is the unique choice which produces a flat connection and, in turn, SAM and OAM operators which satisfy the angular momentum commutation relations.

Importantly, this derivation of the massive SAM-OAM operators gives a very clear picture of the singularity that occurs in the massless limit that prevents such an SAM-OAM decomposition for massless particles. We found that the boost and rotation connections $\DK$ and $\DR$ are anti-Hermitian connections which are naturally defined on particle bundles of any mass. Our method of deriving the flat connection for massive particles relied on the fact that these are two different connections with different curvatures, and thus by taking affine sums we could produce a family of connections with different curvatures. Among this family was a unique globally well-defined flat connection. However, it is a strange fact of relativity that rotations and boosts orthogonal to the momentum degenerate when $m=0$. In particular, we have (\cite{Weinberg1995}, Eq. (2.5.38); \cite{PalmerducaQin_PT}, Eq. (129))
\begin{align}
    \svop{J}_\perp &= -\phats \times \svop{K} \\
    \svop{K}_\perp &= \phats \times \svop{J}
\end{align}
where 
\begin{align}
        \svop{J}_\perp &\doteq \svop{J} - \phats(\phats \cdot \svop{J}) \\
        \svop{K}_\perp &\doteq \svop{K} - \phats(\phats \cdot \svop{K}).
\end{align}
This is a nonobvious relationship which originates from imposing the condition that massless particles have a finite number of internal degrees of freedom. It would not hold for massless particles with continuous spin, but such particles have never been observed and are not part of the Standard Model \cite{Weinberg1995}. This condition says that, infinitesimally, boosting in some direction $\vhat$ perpendicular to $\khat$ is equivalent to rotating in the plane containing $\vhat$ and $\khat$. Thus, parallel transporting a vector via boosts or rotations will give the same result. Indeed, for massless particles we have $\sop{H} = |\svop{P}|$ and
\begin{equation}
    \svop{K} = \phats(\phats \cdot \svop{K})+ \phats \times \svop{J}
\end{equation}
as shown in Eq. (\ref{eq:K_decomp}). Plugging these into Eq. (\ref{eq:DK_b}) and comparing with Eq. (\ref{eq:DR_1}) shows that the boost and rotation connections degenerate for massless particles:
\begin{equation}
    \DK = \DR.
\end{equation}
There is now a single non-flat connection induced by \Poincare symmetry and one cannot produce a non-flat connection via an affine combination of $\DK$ and $\DR$. Thus, the lack of an SAM-OAM decomposition for massless particles can be traced back to the degeneracy of transverse boosts and rotations.

However, this also shows that there is a theoretical origin of the massless splitting in Eq. (\ref{eq:massless_splitting}). The most naive splitting of the massless angular momentum leads to operators which are not well-defined as they violate the transversality of EM waves \cite{Akhiezer1965}. The $\mbf{J}_\parallel$, $\mbf{J}_{\perp}$ split was an ad hoc fix, solving the transversality issue, but leading to operators which are not legitimate SAM-OAM operators. However, here we show see that these operators naturally arise as the unique splitting induced by Poincare symmetry, and belongs to the same class of splittings as the well-defined SAM-OAM splitting for massive particles.

\section{Conclusion}
We have shown that any rotationally symmetric connection on particle bundles induces a splitting of the angular momentum operator into two parts $\SDbold$ and $\LDbold$. $\SDbold$ is always internal and $\LDbold$ has the form $-i\kvec \times \cx$, and in this sense they \emph{resemble} spin and orbital angular momentum operators. However, the defining feature of angular momentum operators is that they satisfy angular momentum commutation relations, and we find that $\SDbold$ and $\LDbold$ do so if and only if the $S^2$-curvature of $\mbf{D}$ vanishes. \Poincare symmetry induces two connection, corresponding to parallel transport via boosts or rotations. For massive particles these are different connections and there is a unique affine combination of them which produces a well-defined SAM-OAM splitting. The result agrees with the SAM-OAM operators for relativistic massive particles resulting from the Newton-Wigner position operator. The method presented arguably gives a more intuitive derivation of these rather complicated operators. For massless particles, however, we find that the connections $\DR$ and $\DK$ degenerate, and it is not possible to combine them to produce a flat connection. This gives a geometric picture of why the massless limit obstructs an SAM-OAM splitting for massless particles such as photons and gravitons. Nevertheless, the degeneracy of $\DR$ and $\DK$ means that for massless particles there is a unique splitting of $\mbf{J}$ induced by \Poincare symmetry, namely the $\mbf{J}_\parallel$, $\mbf{J}_{\perp}$ splitting. This splitting has been previously suggested as an ad hoc fix for the SAM-OAM decomposition problem for photons, but it was difficult to interpret or justify since neither $\mbf{J}_\parallel$ nor $\mbf{J}_{\perp}$ are actually angular momentum operators. Here we show the theoretical origins of these operators.

We suspect that some of our findings can be extended to electromagnetic waves in plasmas and other media, as we will describe heuristically. Introducing a homogeneous isotropic medium introduces a preferred frame, namely the rest frame of the medium. One can always solve for the wave solutions in this rest frame, and doing so removes the boost symmetry of the solution space. The solutions will then only have a rotational $\SO(3)$ and translational $\Real^{3+1}$ spacetime symmetry, regardless of whether one chooses to work relativistically or non-relativistically. Thus, the only symmetry-induced connection is the rotation connection which, for transverse waves, is curved. As such, the rotation connection will induce a splitting of the total angular momentum into operators which are not true angular momentum operators. By this reasoning, we conjecture that most homogeneous isotropic media, e.g. plasmas, will not allow an SAM-OAM decomposition for EM waves. Rigorous analysis is needed to verify this.

\begin{acknowledgments}
This work is supported by U.S. Department of Energy (DE-AC02-09CH11466).
\end{acknowledgments}

\appendix
\section{Noncommutativity of the \Poincare generators}\label{app:Appendix}
The \Poincare generators satisfy
\begin{subequations}\label{eq:Poincare_algebra}
\begin{align}
    [J_a,J_b] &= i\epsilon_{abc}J_c \,, \;\;  [J_a, K_b] = i\epsilon_{abc} K_c \\
    [K_a,K_b] &= -i\epsilon_{abc}J_c \,, \;\; [J_a,P_b] = i\epsilon_{abc}P_c \\
    [K_a,P_b] &= iH\delta_{ab} \,, \;\; [K_a,H] = iP_a \\
    [J_a,H] &= [P_a,H] = [P_a,P_b] = [H,H] = 0.
\end{align}
\end{subequations}
In this appendix we give some useful results pertaining to the noncommutativity of the generators.

\subsection{Noncommutative BAC-ABC rule}
Ordinary 3D vectors $\mbf{a},\mbf{b},\mbf{c}$ satisfy the well-known BAC-CAB rule
\begin{equation}
    \mbf{a} \times (\mbf{b} \times \mbf{c}) = \mbf{b}(\mbf{a}\cdot \mbf{c}) - \mbf{c}(\mbf{a}\cdot \mbf{b}).
\end{equation}
However, this relationship does not hold for general vector operators $\mbf{A}$, $\mbf{B}$, $\mbf{C}$ which fail to commute. If $\mbf{A}$ and $\mbf{B}$ commute, then we have the following useful relationship
\begin{proposition}[Noncommutative BAC-ABC relation]
    Let $\mbf{A}$, $\mbf{B}$, $\mbf{C}$ be vector operators such that $\mbf{A}$ and $\mbf{B}$ commute: $[A_m,B_n]=0$. Then
    \begin{equation}
        \mbf{A} \times (\mbf{B} \times \mbf{C}) = \mbf{B}(\mbf{A} \cdot \mbf{C}) - (\mbf{A} \cdot \mbf{B})\mbf{C}
    \end{equation}
\end{proposition}
\begin{proof}
    From the relation 
    \begin{equation}
        \epsilon_{ijk}\epsilon_{mnk} = \delta_{im}\delta_{jn} - \delta_{in}\delta_{jm}
    \end{equation}
    we obtain
    \begin{align}
        [\mbf{A} \times (\mbf{B} \times \mbf{C})]_i &= \epsilon_{ijk} A_j\epsilon_{kmn}B_mC_n \\
        &= (\delta_{im}\delta_{jn} - \delta_{in}\delta_{jm})A_jB_mC_n \\
        &=A_jB_iC_j - A_jB_jC_i \\
        &= \mbf{B}(\mbf{A} \cdot \mbf{C}) - (\mbf{A}\cdot\mbf{B})\mbf{C}
    \end{align}
\end{proof}
This is useful for decomposing a vector operator, such as $\svop{J}$, into parallel and perpendicular parts with respect to another such as $\phats$. In particular, since $\phats \cdot \phats = 1$, applying the above result with $\mbf{A} = \mbf{B} = \phats$ and $\mbf{C}= \svop{J}$ gives
\begin{align}
    \svop{J} &= \phats(\phats \cdot \svop{J}) - \phats \times (\phats \times \svop{J}) \\
    &\doteq \svop{J}_\parallel + \svop{J}_\perp
\end{align}
where
\begin{subequations}
\begin{align}
    \svop{J}_\parallel &\doteq \phats(\phats \cdot \svop{J}) \\
    \svop{J}_\perp &\doteq - \phats \times (\phats \times \svop{J}).
\end{align}
\end{subequations}
Similarly,
\begin{subequations}\label{eq:K_decomp}
\begin{align}
    \svop{K} &= \svop{K}_\parallel + \svop{K}_\perp \\
    \svop{K}_\parallel &\doteq \phats(\phats \cdot \svop{K}) \\
    \svop{K}_\perp &\doteq - \phats \times (\phats \times \svop{K}).    
\end{align}
\end{subequations}

\subsection{Useful commutation relations}
The boost generators $\svop{K}$ satisfy more complicated commutation relations than the other \Poincare generators. Indeed, unlike the momentum operators, the boosts do not commute with $H$, nor do different components of the boost commute with each other. This gives rise to interesting physics, such as the Wigner rotation, but it also greatly complicates calculations. In this section we collect a number of useful commutation relations involving the \Poincare generators, mostly involving $\svop{K}$. 

If the operator $B^{-1}$ is defined, then
\begin{align}
    [A,B^{-1}] &= AB^{-1} - B^{-1}A \nonumber\\
    &= B^{-1}BAB^{-1} - B^{-1}ABB^{-1} \nonumber\\
    &= -B^{-1}[A,B]B^{-1}.\label{eq:inverse_comm}
\end{align}
If $\big[[A,B],B\big] = 0$, then for non-negative $n$: (\cite{Zettili2009}, Eq. 2.90)
\begin{equation}
    [A,B^n] = n[A,B]B^{n-1}. \label{eq:power_comm}
\end{equation}
If $B^{-1}$ exists, then Eq. (\ref{eq:power_comm}) holds for all integers $n$ by Eq. (\ref{eq:inverse_comm}). Using the \Poincare commutation relations and Eq. (\ref{eq:power_comm}) we obtain the following:
\begin{subequations}
\begin{align}
    &[K_a,H^n] = inP_aH^{n-1} \label{eq:comm_Ka_hn}\\
    &[K_a,|\mbf{P}|^n] = inHP_a|\mbf{P}|^{n-2} \label{eq:comm_Ka_P}\\
    &[(\mbf{P}\cdot \mbf{K}),P_a] = iHP_a  \\
    &[\hat{P}_a, K_b] = -\frac{i\delta_{ab}H}{|\mbf{P}|} + \frac{iHP_aP_b}{|\mbf{P}|^3} \\
    &[H^mK_a, H^nP_b] = iH^{m+n+1}\delta_{ab} + inP_aP_bH^{m+n-1} \label{eq:comm_id_hm_k_hn_p} \\
    &\hat{\mbf{P}}\cdot \mbf{K} - \mbf{K}\cdot \hat{\mbf{P}} = -\frac{2iH}{|\mbf{P}|} \\
    &\mbf{P} \cdot \mbf{K} - \mbf{K} \cdot\mbf{P} = -3iH.
\end{align}
\end{subequations}

\bibliography{connection}
\end{document}